\documentclass[a4paper, conference]{IEEEtran}
\usepackage[hidelinks]{hyperref} 
\usepackage{cleveref}
\usepackage{comment}
\usepackage{caption} 
\usepackage{dsfont}
\usepackage{amsthm}
\usepackage{amsmath}
\usepackage{amssymb}
\usepackage{lipsum}
\usepackage{esint}
\usepackage{graphicx, accents}
\usepackage{mathrsfs}             
\usepackage{amssymb}
\usepackage{epstopdf}
\usepackage{textpos}
\usepackage{lipsum}
\usepackage{color}
\usepackage{blindtext}
\usepackage{braket}
\usepackage{bbm}

\newtheorem{thm}{Theorem}\crefname{thm}{Theorem}{Theorems}
\crefname{lem}{Lemma}{Lemmas}
\crefname{prp}{Proposition}{Propositions}
\crefname{cor}{Corollary}{Corollaries}
\crefname{dfn}{Definition}{Definitions}
\crefname{section}{Section}{Sections}
\crefname{appendix}{Appendix}{Appendices}
\crefname{eq}{Equation}{Equations}

\DeclareMathOperator{\Prob}{Prob}
\DeclareMathOperator{\tr}{tr}
\DeclareMathOperator{\Var}{Var}
\newcommand{\R}{\mathbb{R}}

\newcommand{\C}{\mathbb{C}}

\newcommand{\E}{\mathbb{E}}
\newcommand{\state}[1]{{\ket{#1}\bra{#1}}}
\newcommand{\abs}[1]{\vert{#1}\vert}

\newcommand{\ot}{\otimes}

\newcommand{\nth}{n_\mathrm{th}}
\newcommand{\rth}{\rho_\mathrm{th}}
\renewcommand{\Re}{\mathrm{Re}}
\renewcommand{\Im}{\mathrm{Im}}

\begin{document}
\title{Noisy Feedback and Loss Unlimited \\Private Communication}
\author{
\IEEEauthorblockN{Dawei Ding}
\IEEEauthorblockA{Stanford Institute for Theoretical Physics\\
Stanford University\\
Stanford, CA 94305, USA \\
Email: dding@stanford.edu}
\and \IEEEauthorblockN{Saikat Guha}
\IEEEauthorblockA{College of Optical Sciences\\
University of Arizona\\
Tucson, AZ 85721, USA \\
Email: saikat@optics.arizona.edu}
}
\maketitle

\begin{abstract}
Alice is transmitting a private message to Bob across a bosonic wiretap channel with the help of a public feedback channel to which all parties, including the fully-quantum equipped Eve, have completely noiseless access. We find that by altering the model such that Eve's copy of the initial round of feedback is corrupted by an iota of noise, one step towards physical relevance, the capacity can be increased dramatically. It is known that the private capacity with respect to the original model for a pure-loss bosonic channel is at most $- \log(1-\eta)$ bits per mode, where $\eta$ is the transmissivity, in the limit of infinite input photon number. This is a very pessimistic result as there is a finite rate limit even with an arbitrarily large number of input photons. We refer to this as a loss limited rate. However, in our altered model we find that we can achieve a rate of $(1/2) \log(1 + 4 \eta N_S)$ bits per mode, where $N_S$ is the input photon number. This rate diverges with $N_S$, in sharp contrast to the result for the original model. This suggests that physical considerations behind the eavesdropping model should be taken more seriously, as they can create strong dependencies of the achievable rates on the model. For by a seemingly inconsequential weakening of Eve, we obtain a loss-unlimited rate. Our protocol also works verbatim for arbitrary i.i.d.\ noise (not even necessarily Gaussian) injected by Eve in every round, and even if Eve is given access to copies of the initial transmission and noise. The error probability of the protocol decays super-exponentially with the blocklength.
\end{abstract}


\section{Introduction}

We consider the task of continuous variable (CV) private communication over a lossy bosonic channel with additive noise, with a quantum-limited adversary. 
We consider a forward wiretap channel from Alice to Bob with an eavesdropper Eve. There is also a backward 
public side channel from Bob to Alice. All parties, including Eve, have perfectly noiseless access to communication on this side channel. This immediately presents a problem. All communication channels in reality have a non-zero level of noise. With error-correction the noise can be reduced, but never completely removed. Therefore this model is clearly unphysical. Indeed, most of the information theoretic capacity analyses of private communication and secret-key generation have traditionally assumed such a zero-error feedback. Admittedly, the process of abstraction that allows a problem to be mathematically analyzed inevitably causes such unphysical features to emerge. However, we will find in this paper that this particular feature can be problematic in that ameliorating it can create a significant difference in achievable rates.

More explicitly, we propose a protocol over a lossy bosonic channel and consider an alteration of the feedback-assisted private communication model in which Eve's copy of the public communication sent by Bob to Alice on just the initial round of the protocol is corrupted by a small amount of noise, a small step towards physical relevance. Eve obtains the remainder of the public communication noiselessly, and is otherwise quantum equipped, i.e., has perfect quantum memories and the ability to make arbitrary collective quantum measurements. We show that with such a seemingly inconsequential weakening of the eavesdropper, Alice and Bob can achieve a rate of $\frac{1}{2}\log(1+4\eta N_S)$ secure bits per mode~\footnote{Throughout this paper, logarithms are base 2.} using a simple laser-light modulation and homodyne detection, $N_S$ being the mean input photon number per mode. This is a loss unlimited rate, in sharp contrast to the loss limited upper bound known for the original model: $-\log(1-\eta)$~\cite{takeoka2015fundamental, pirandola2017fundamental}. 
The exact same protocol also works if the channel has noise in addition to loss. We will argue that the noise can even be non-Gaussian.

A closely related classical result~\footnote{Note that in the classical case the feedback-assisted private capacity of the Gaussian wiretap channel also has a loss limited capacity. This can be seen for instance by upper bounding it by the secret key capacity and upper bounding that via Theorem 4 of~\cite{maurer1993secret}.} was proven in~\cite{gunduz2008secret}. However, their altered communication model assumes Eve's copy of the feedback is corrupted by \textit{Gaussian additive noise} and on \textit{every} round. Thus this is a more specialized and significant departure from the usual model. They do consider more general correlated noise between the different channels involved, but this can be integrated into our model as well. 
In addition to the difference in model alteration, in the discussion section we will see that we can slightly \textit{strengthen} Eve from what is discussed above, by giving her copies of the initial transmission and noise, and still obtain our main result with the same protocol. Note that our results trivially apply to a classical additive white Gaussian noise (AWGN) channel and since adding more noise to Eve's feedback can only increase the achievable rate, implies the result in~\cite{gunduz2008secret}.

\section{Main Result}
Alice is to send a message to Bob over a bosonic (quantum) wiretap channel ${\cal N}$, which is to be kept private from Eve, by transmitting bosonic states on the forward channel $\cal N$ and using a backward noiseless classical feedback channel. Eve obtains the quantum output of the complement of the isometric extension of ${\cal N}$, and can eavesdrop on the feedback channel, but she obtains a noisy version of the classical feedback for the initial round. Denoting the classical feedback system as $Y$, for this round Eve obtains the output of a classical noisy channel $W(Y)$. 
Furthermore, we assume that the capacity of the noisy channel $W(Y)$ from Bob to Eve, is finite for a finite input power, that is,
\begin{equation*}
  C(W, P) \equiv \max_{Y: \E[Y^2] \le P} I(Y; W)
\end{equation*}
is finite for finite~\footnote{Note that the capacity is finite at all finite cost constraints iff it is finite at some finite cost constraint. The forward direction is trivial. For the reverse direction, assume for contradiction that there is some finite cost constraint at which the capacity is infinite. Then, by time-sharing we can conclude that is infinite at all finite cost constraints, a contradiction.} $P$. Note that this condition is very mild. For instance, if $W(Y) = Y + S$, with $S \sim {\rm Gaussian}(0,N_0)$, is AWGN, $N_0$ can be arbitrarily small as long as it is nonzero. Hence, we only need an iota of noise.
 We shall refer to this setting as \emph{asymmetric feedback-assisted private communication.} We say asymmetric because the feedback is noiseless to Alice but noisy to Eve on the initial round. A diagram of the communication model for the initial round is shown in~\cref{fig:model}. $B'$ is some leftover system that Bob keeps. For subsequent rounds no noise is applied to Eve's copy of $Y$.
\begin{figure}[h]
  \centering
  \includegraphics[width=0.7\columnwidth]{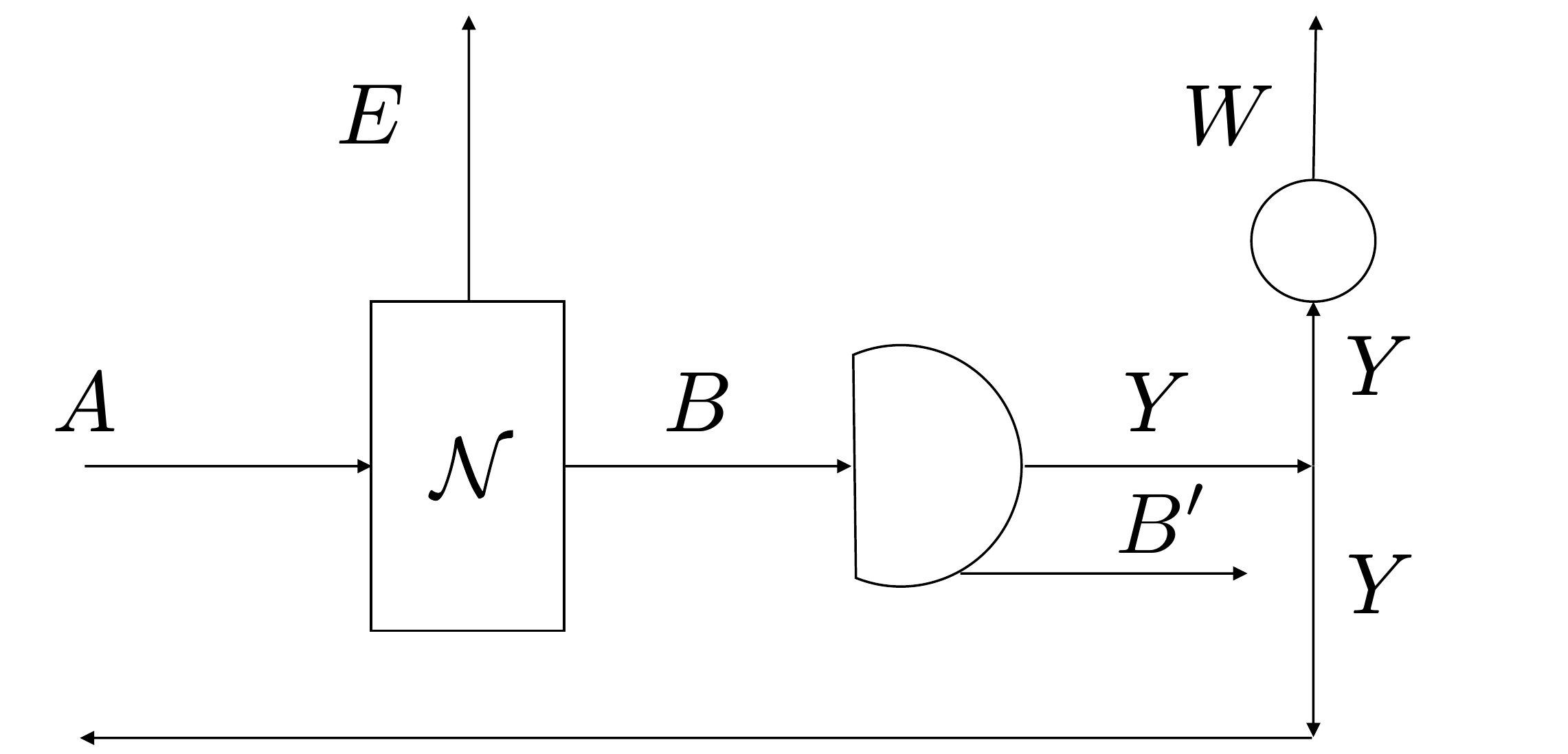}
  \caption{A diagram of the asymmetric feedback-assisted private communication model for the initial round of communication.}
  \label{fig:model}
\end{figure}

The privacy criterion we use is the same as the one used in definitions of private communication in various classic papers~\cite{wyner1975wire,csiszar1978broadcast,maurer1993secret} and in the classical result~\cite{gunduz2008secret}; i.e., we assume the message $M$ is uniformly distributed and require: 
  \begin{equation*}
    \lim_{k \to \infty} \frac{1}{k} I(M; \mathcal{E}) = 0,
  \end{equation*}
where $k$ is the number of rounds and $\mathcal{E}$ is the aggregation of Eve's systems after all $k$ rounds. Note that this is in general weaker than the privacy criterion used in quantum Shannon theory analyses, for instance~\cite{wilde2013quantum}, which requires the trace distance between Eve's states and a constant state independent of the message to vanish. We leave open whether we can obtain our result with this stronger privacy criterion.

We will consider the case where the forward channel is a single mode thermal bosonic wiretap channel: denoting Alice's, Bob's and Eve's optical modes as $\hat a, \hat b,$ and $\hat e$, respectively, the channel is described by the Heisenberg evolutions:
\begin{equation}
\hat b = \sqrt{\eta} \hat a + \sqrt{1 - \eta} \hat f; \qquad  \hat e = \sqrt{1-\eta} \hat a - \sqrt{\eta} \hat f,
\end{equation}
where $\hat f$ is the optical mode of the environment, which is in a thermal state ${\rho_{\rm th}}(0, n_{\rm th})$ with thermal number $\nth$, and $\eta$ is the transmissivity. A thermal state is given by $\rho_\text{th}(\alpha, \bar n) \equiv \int_{\C}^{} \frac{d^2 \beta}{\pi \bar n} e^{-| \beta |^2 / \bar n} \state{\alpha + \beta}.$ We enforce an average photon number constraint on Alice's input over the $i \in [0:n]$ rounds: 
$\frac{1}{n+1} \sum_{i=0}^{n} \tr[\hat a^\dagger \hat a \rho_{A_i}] \le N_S$.
We now state our main theorem.
\begin{thm}
  \label{thm:main}
  The rate
  \begin{equation}
    \label{eq:shannonHom}
    P_H = \frac{1}{2}\log \left (1+ \frac{N_S}{\sigma^2}\right),
  \end{equation}
  where  $\sigma^2 \equiv \frac{1}{4\eta} + \frac{1-\eta}{2\eta} \nth$, is achievable with mean photon per mode $N_S$ for asymmetric feedback-assisted private communication over the thermal bosonic wiretap channel, under the aforesaid model. 
\end{thm}

\begin{proof}
  The rate is achieved by generalizing the Schalkwijk-Kailath protocol~\cite{schalkwijk1966coding} to the quantum setting.

  \textbf{Codebook}:
  Divide the interval $[- \sqrt{N_S}, \sqrt{N_S}]$ into $2^{nR}$ equal length intervals for some $R \in \R^+$. Given a message $m \in [2^{nR}]$, let $\theta(m)$ be the midpoint of the $m$-th interval. The codebook is then given by the coherent states $\ket{\theta(m)}$.

  \textbf{Encoding and Decoding}:
When we say Alice transmits a random variable $X$ on $\R$, we mean she sends the mixed state 
  \begin{equation*}
    \rho_X \equiv \int_\R dx p_X(x) \state {x}, 
  \end{equation*}
  where $\ket{x}$ is a coherent state. Note that $\E[X^2] = \tr[\hat a^\dagger \hat a \rho_X]$.
  Given message $m$, on round $i=0$, Alice sends the constant random variable $X_0$ which takes the value $\theta(m)$. Bob receives $\rho_\text{th}(\sqrt{\eta} \, \theta(m), (1-\eta) \nth)_{B_0}$, and Eve receives $\rho_\text{th}(\sqrt{1-\eta} \, \theta(m), \eta \nth)_{E_0}$.
  Bob then performs an $x$ homodyne measurement on his output, to obtain, with appropriate scaling, the output of an induced additive white Gaussian noise (AWGN) channel with input $X_0$:  $Y_0 \sim \mathcal{N}(\theta(m), \sigma^2)$. Now, on every round, Bob will simply send a copy of his measurement result through the feedback channel. Via the feedback channel Alice receives $Y_0$, while Eve receives $W_0(Y_0)$. 
 
  Denote the random noise that is added in the $i$-th induced AWGN channel from Alice to Bob by $N_i \equiv Y_i - X_i$. Note that since Alice has access to both $Y_i$ and $X_i$, she can always explicitly compute $N_i$. For round $i\ge 1$, Alice computes the conditional expectation $\E[N_0 \vert Y^{i-1}]$, where $Y^{i-1} \equiv Y_1 \dots Y_{i-1}$, that is, $Y_0$ is \textit{not} included, and $\E[N_0 \vert Y^0] \equiv \E[N_0]= 0$. She then transmits 
  \begin{equation*}
    X_i = \gamma_i (N_0 - \E[N_0 \vert Y^{i-1}]),
  \end{equation*}
  where $\gamma_i$ are chosen such that $\E[X_i^2] = N_S$. Note that this guarantees that the photon number constraint is satisfied and with equality. Bob receives the state
  \begin{equation*}
    \int_\R dx p_{X_i}(x) \rho_\text{th}(\sqrt{\eta} x, (1-\eta) \nth)_{B_i},
  \end{equation*}
  on which he will do a homodyne $x$ measurement to obtain $Y_i$, while Eve receives
  \begin{equation*}
    \int_\R dx p_{X_i}(x) \rth(\sqrt{1-\eta}x, \eta \nth)_{E_i}.
  \end{equation*}
  At the end of the protocol ($i =n$), Bob computes
  \begin{equation*}
    \Theta_n \equiv Y_0 - \E[N_0 \vert Y^n] = \theta(m) + (N_0 - \E[N_0 \vert Y^n]).
  \end{equation*}
  His guess of the message will then be $\hat m$, where $\theta(\hat m)$ is closest to $\Theta_n$.
  \\~\\
  \textbf{Error Analysis}:
  Note that this protocol is simply the Schalkwijk-Kailath protocol over an induced classical AWGN channel with noise $\sigma^2$. Therefore, by the analysis in~\cite{schalkwijk1966coding,el2011network}, the probability of error of this protocol is upper bounded by
  \begin{equation*}
    p_e \le \sqrt{\frac{2}{\pi}} \exp\left(- 2^{2n (P_H - R) -1} N_S/\sigma^2 \right).
  \end{equation*}
  Hence, $p_e \xrightarrow[n \to \infty]{} 0$  if $R < P_H$. For completeness, we reproduce this argument in~\cref{appx:SK}.
  \\~\\
  \textbf{Privacy Analysis}:
For round $i \ge 1$, the message is no longer involved in the protocol as Alice is simply trying to transmit the noise $N_0$, which is independent of the message $M$, to Bob. Hence we can upper bound the mutual information between the message and Eve's systems as follows:
  \begin{eqnarray}
    && I(M ; E_0^n W_0^n) \nonumber \\
    &\le& I(M ; E_0^n W_0^n N_0) \nonumber \\
    &=& I(M; E_0 W_0 N_0) + I(M; E^n W^n | E_0 W_0 N_0)\nonumber \\
    &=& I(M; E_0 W_0 N_0) \label{eq:privacyBound}.
  \end{eqnarray}
  The inequality follows from monotonicity of mutual information. The last equality follows from our observation that the subsequent transmissions are only about $N_0$ and independent of the message. 
  We explicitly compute the state at the end of round $i=0$. 
  We start with a uniformly distributed message:
  \begin{equation*}
    \frac{1}{2^{nR}} \sum_{m=1}^{2^{nR}} \state{m}_M.
  \end{equation*}
  Alice encodes the message into a coherent state:
  \begin{equation*}
    \frac{1}{2^{nR}} \sum_{m=1}^{2^{nR}} \state{m}_M  \ot \state{\theta(m)}_{A_0}.
  \end{equation*}
  $A_0$ is then sent into the thermal bosonic wiretap channel:
  \begin{align*}
    & \frac{1}{2^{nR}} \sum_{m=1}^{2^{nR}} \state{m}_M   \ot \int_{\C}^{}  \frac{d^2 \alpha}{\pi \nth} e^{-\abs{\alpha}^2/\nth} \\
    & \state{\sqrt{\eta} \theta(m) + \sqrt{1-\eta} \alpha}_{B_0}\\
    & \ot \state{\sqrt{1-\eta} \theta(m) - \sqrt{\eta} \alpha}_{E_0}.
  \end{align*}
  Bob then does a homodyne $x$ measurement, which produces the system $Y_0, N_0$ (None of the parties have posession of $N_0$.):
  \begin{align*}
    &\frac{1}{2^{nR}} \sum_{m=1}^{2^{nR}} \state{m}_M  \ot \int_{\C}^{} \frac{d^2 \alpha}{\pi\nth} e^{-\abs{\alpha}^2/\nth} \int_\R dy_0 \sqrt{\frac{2\eta}{\pi}} \\
    & e^{-2\eta (y_0 - \theta(m) - \sqrt{\frac{1-\eta}{\eta}} \Re(\alpha))^2} \state{y_0}_{Y_0} \\
    & \ot \state{y_0 - \theta(m)}_{N_0} \\
    & \ot \state{\sqrt{1-\eta} \theta(m) - \sqrt{\eta} \alpha}_{E_0}.
  \end{align*}
  Note that $\state{y_0}, \state{y_0 - \theta(m)}$ are classical states and are not coherent states. Bob sends a noisy copy to Eve to obtain the state:
  \begin{align}
    & \frac{1}{2^{nR}} \sum_{m=1}^{2^{nR}} \state{m}_M  \ot \int_{\C}^{} \frac{d^2 \alpha}{\pi \nth} e^{-\abs{\alpha}^2 /\nth} \int_\R d n_0 \sqrt{\frac{2\eta}{\pi }} \nonumber \\
    & e^{- 2\eta (n_0 - \sqrt{\frac{1-\eta}{\eta}}\Re(\alpha))^2 }  \state{n_0}_{N_0}  \nonumber\\
    & \ot \int_\R d w_0 p_{W_0 | Y_0} {(w_0|\theta(m) + n_0)} \state{w_0}_{W_0}\nonumber \\
    & \ot \state{\sqrt{1-\eta} \theta(m) - \sqrt{\eta} \alpha}_{E_0}, \label{eq:stateZero}
  \end{align}
  where $p_{W_0 | Y_0}$ is the probability distribution of the output of the noisy feedback channel given input $Y_0$ and we made the substitution $n_0 \equiv y_0 -\theta(m)$. 

  Using \eqref{eq:stateZero}, we can bound 
  \begin{eqnarray}
    &&I(M; E_0 W_0 N_0) \nonumber \\
    &=& I(M ; E_0 W_0 | N_0) + I(M; N_0) \nonumber\\
    &\le&  I(M Y_0; E_0 W_0) \nonumber \\
    &\le&  I(M ; W_0 | Y_0) + I(Y_0; W_0) + S( E_0 | W_0) \nonumber \\
    &\le&  I(Y_0; W_0) + S(E_0) \label{eq:privacyBound2},
  \end{eqnarray}
  where we used the facts that $M$ and $N_0$ are independent, $N_0 = Y_0 - X_0$, $M \to Y_0 \to W_0$ form a Markov chain, and that $S(E_0 \vert M Y_0 W_0) \ge 0$ since the conditioned systems are classical. Now, the first term is simply the mutual information between the input and output of the noisy channel to Eve on the initial round. Hence, it is upper bounded by the capacity $C(W, N_S + \sigma^2)$ of feedback channel from Bob to Eve with cost constraint $N_S + \sigma^2 = \E[(Y_0)^2]$: 
  \begin{align}
    I(Y_0; W_0) 
    & \le C(W, N_S + \sigma^2) \label{eq:Szero}.
  \end{align}
  This capacity is finite by assumption. We can bound the second term in~\eqref{eq:privacyBound2} by observing
  \begin{align*}
    & \tr[({\hat e}^\dagger {\hat e}) \rho_{E_0}] = \frac{1}{2^{nR}} \sum_{m=1}^{2^{nR}} \int_{\C}^{} \frac{d^2 \alpha}{\pi \nth} e^{-\abs{\alpha}^2/\nth} \\
    & \left[ \left( \sqrt{1-\eta}\theta(m) - \sqrt{\eta}\,\Re(\alpha) \right)^2 + \eta \Im(\alpha)^2 \right]\\
    & \le \frac{1}{2^{nR}} \sum_{m=1}^{2^{nR}} \int_{\C}^{} \frac{d^2 \alpha}{\pi \nth} e^{-\abs{\alpha}^2/\nth} \left[ (1-\eta)N_S + \eta \abs{\alpha}^2 \right]\\
    & = (1-\eta) N_S + \eta \nth,
  \end{align*}
  where in the inequality the term linear in $\Re(\alpha)$ vanishes after integrating over $\alpha$. Hence, $S(E_0) \le g(  (1-\eta) N_S + \eta \nth)$ where $g(x) \equiv (x+1) \log(x+1) - x \log x$ \cite{caves1994quantum,bekenstein1981universal}. 
  
  Since both terms of independent of $n$, we conclude
  \begin{align*}
    & \lim_{n+1 \to \infty} \frac{1}{n+1} I(M; E_0^n W_0^n)\\
    & \le \lim_{n+1 \to \infty} \frac{1}{n+1} I(M; E_0 W_0 N_0) \\
    & \le \lim_{n+1 \to \infty} \frac{1}{n+1}(I(Y_0; W_0) + S(E_0)) \\
    & \le \lim_{n+1 \to \infty} \frac{C(W,N_S + \sigma^2) + g((1-\eta) N_S + \eta \nth )}{n+1} \\
    & = 0,
  \end{align*}
  thereby establishing privacy.
\end{proof}

\section{Discussion}
We have proven that with respect to the communication model of asymmetric feedback where Eve obtains a slightly noisy copy of the feedback on the initial round of the protocol, it is possible to achieve a much higher rate of private communication over a lossy thermal-noise bosonic wiretap channel, even one that is loss unlimited~\eqref{eq:shannonHom}. 
It is worthwhile to make some observations about this result, as well as discuss its implications, limitations, and possible extensions. 

\subsection{Extensions}
Eq.\ \eqref{eq:shannonHom} is achievable with a probability of decoding error that decays doubly exponentially with the number of channel uses. That is, the rate is achievable with an infinite error exponent. Indeed, this was the main thrust of the result in~\cite{schalkwijk1966coding} for the setting of feedback-assisted classical communication. It was then extended by~\cite{gallager2010variations}, where it was shown that we can actually achieve a probability of error that goes as an exponential tower with an order that increases linearly with the blocklength. This was shown to be optimal by a corresponding lower bound on the error probability. This lower bound would also apply in our case since it only concerns the part where Alice reliably sends a message to Bob. The only essential difference between the protocol in~\cite{gallager2010variations} and that of Schalkwijk and Kailath is that the former preserves the discrete structure of the message. In particular, in both, only the initial transmission contains information about the message. It is then straightforward to make the same extension to our setting. 
That is, we can achieve~\eqref{eq:shannonHom} with a block probability of error
\begin{equation*}
  p_e \le \frac{1}{^{f(n)} e}
\end{equation*}
for sufficiently large $n$, where $^n a$ denotes the tetration operation,
\begin{equation*}
  f(n) \equiv \left\lfloor n[1-\varphi^{-1}(R)] - \frac{5(1-\varphi^{-1}(R))}{P_H - R} \right \rfloor,
\end{equation*}
and
\begin{equation*}
  \varphi(\nu) \equiv \frac{\nu}{2} \log \left(1 + \frac{N_S}{\sigma^2 \nu} \right).
\end{equation*}
Note that this expression has minor differences compared to the bound in~\cite{gallager2010variations} because we use bits instead of nats. Some simplifications were also made for presentation.

Another possible extension follows from observing that the achieved rate~\eqref{eq:shannonHom} is not optimal. This is because it is known that squeezed state encoding can achieve higher rates over the thermal noise lossy bosonic channel with a homodyne receiver~\cite{guha2004classical}. However, the induced classical channel with squeezed states encoding is AWGN, so our argument trivially applies to this case. With optimal squeezed-state encoding, we can achieve a rate~\cite{guha2004classical}:
  \begin{equation}
    \label{eq:shannonsqueezedHom}
    P_{\text{sq}} = \frac{1}{2}\log \left(1+ \frac{4N_S + 2 - f(\eta, N_S) + f(\eta, N_S)^{-1}}{((1-\eta)/\eta) + f(\eta, N_S)^{-1}}\right),
  \end{equation}
  where
\begin{equation*}
f(\eta, N_S) \equiv \frac{\eta\left[\left(1 + \frac{2(1-\eta)}{\eta}  \left(\frac{1+\eta}{2\eta}+2N_S\right)\right)^{1/2}-1\right]}{1-\eta}.
\end{equation*}
When Bob uses a homodyne detection receiver, a squeezed state encoding is conjectured to be the optimal encoding for any Gaussian bosonic channel. If this conjecture is proven true, Eq.\ \eqref{eq:shannonsqueezedHom} would be optimal for asymmetric feedback-assisted private communication assuming that Bob makes a homodyne detection since the private capacity is trivially upper bounded by the classical capacity.

\subsection{Non-Gaussian Channels}

An apparent limitation of our result is that it strongly relies on the induced classical channel from Alice to Bob being an AWGN channel. This fact was heavily used in for instance the error analysis of the Schalkwijk-Kailath protocol. However, as noted in Remark 17.3 of~\cite{el2011network}, this exact protocol, that is, the protocol with the conditional expectations $\E[N_0 | Y^{i-1}]$ realized in the AWGN case, can be applied for channels with non-Gaussian additive noise. In fact, it would work for any affine channel whose output is given by
\begin{equation}
  \label{eq:linearChannel}
  a (X + N),
\end{equation}
where $X$ is the input, $N$ is some general additive noise independent of $X$, and $a \neq 0$ is some scaling. The error analysis follows by observing that the estimate $\Theta_n$ is linear in $X_0, N_0^n$ as shown in~\cite{el2011network}. Hence, in the affine channel case the variance of $\Theta_n/a$ has the exact same expression in terms of $X_0, N_0^n$ as in the AWGN case, and so $\Var[\Theta_n] = a^2 \Var[N] 2^{-2n C(\Var[N], N_S)} $, where
\begin{align}
  \label{eq:nonGaussian}
  C(\Var[N], N_S) & \equiv \frac{1}{2} \log \left( 1+ \frac{N_S}{\Var[N]} \right) 
\end{align}
is the capacity of the AWGN channel with the corresponding noise statistics. Then, by Chebyshev's inequality, 
\begin{align*}
  p_e & \le \Prob\left\{ \abs{\Theta_n - \theta(m)} > \sqrt{N_S} 2^{-nR} \right\} \\
  & \le a^2 2^{-2n(C(\Var[N], N_S) - R)} \frac{\Var[N]}{N_S}.
\end{align*}
This works for general affine channels, albeit at the cost of losing the doubly exponential decay. The privacy analysis can also be generalized since the argument did not make use of Gaussianity. 
Note that~\eqref{eq:nonGaussian} is in general less than the actual capacity of the channel, but it is still loss unlimited. In particular, our result applies to a general lossy bosonic wiretap channel where the additive noise (state of the ${\hat f}$ mode) can be any state, that is, it does not even need to be Gaussian. This is shown in~\cref{fig:nonGauss}.
\begin{figure}[h]
  \centering
  \includegraphics[width=0.3\textwidth]{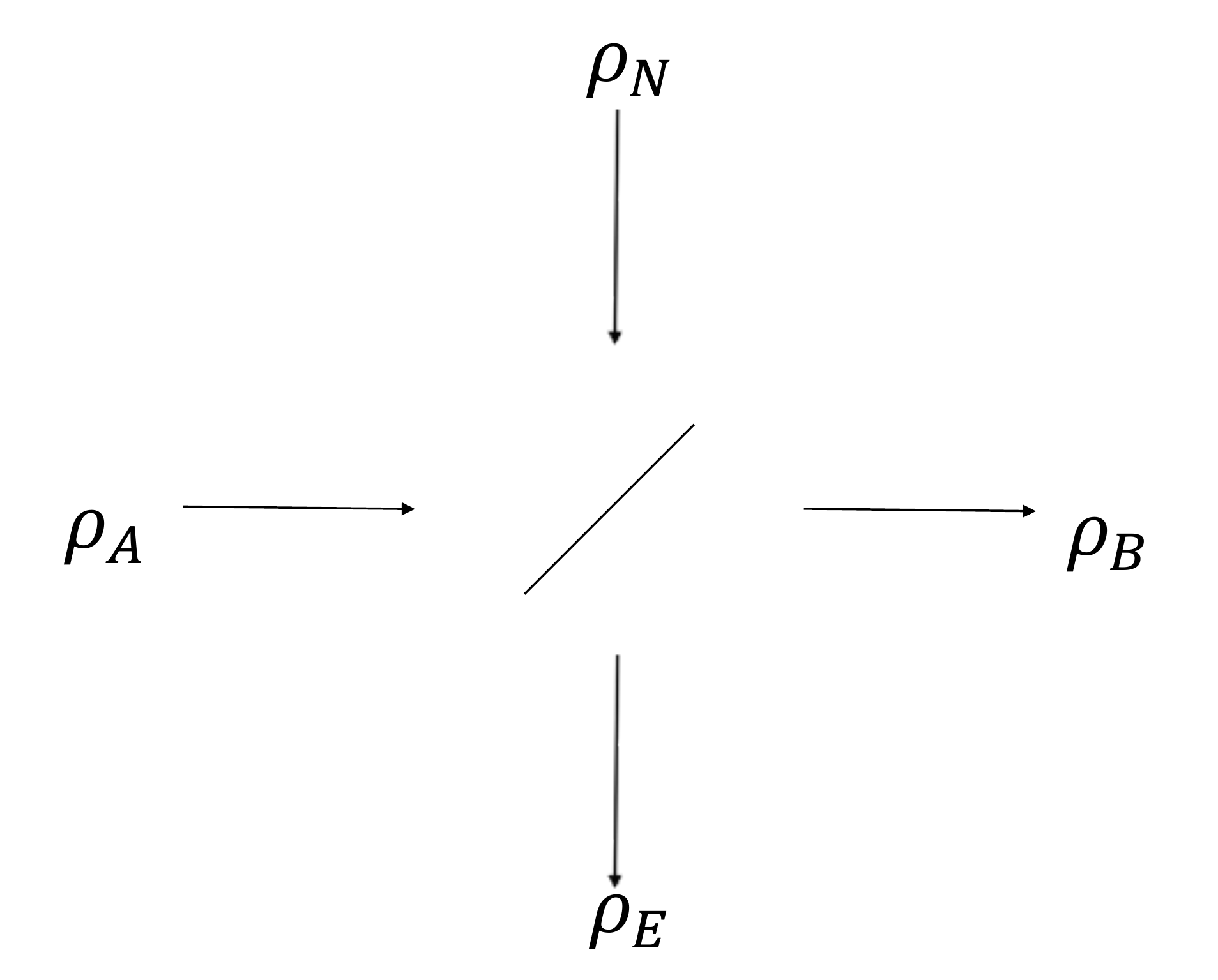}
  \caption{A non-Gaussian channel corresponding to a bosonic wiretap channel where the additive noise $\rho_N$ is not Gaussian.}
  \label{fig:nonGauss}
\end{figure}
It is simple to see that the induced classical channel in this case is an affine channel. This subsumes a large class of physical optical channels.

However, it would be interesting to extend our result to a general quantum wiretap channel. 
Unfortunately, there are channels for which it is not possible to induce a classical affine channel via a particular encoding and decoding. A trivial example is the constant channel which outputs a state regardless of the input. Admittedly, such a channel is useless for communication, so it is not surprising that our protocol fails there. For the specific case of coherent encoding and homodyne detection, we can find explicit nontrivial channels for which we do not induce an affine channel. For instance, consider the channel, a unitary self-Kerr interaction, which maps a coherent state to a cat state: $\ket{\alpha} \mapsto \frac{1}{\sqrt{2}} (\ket{\alpha} + \ket{-\alpha})$. It is not difficult to see that a homodyne measurement of a cat state will always have an expected value of zero. If the channel is affine: $X \mapsto a(X +N)$, then given $\E[ a(X+N)]=0$, either $a= 0$ or $\E[N] = - \E[X]$. If $a =0$, then the measurement result is always zero, which is clearly wrong. The latter possibility is also impossible since additive noise must be independent of the message. The conclusion follows. Thus, it is not clear whether our protocol can be extended to general quantum channels.


\subsection{Communication Model}
We assumed that Eve can eavesdrop on both the (noisy, quantum) forward and (classical) backward channels, but that her copy of the initial round of classical feedback has some noise. From a physical standpoint, this is a natural assumption since there is always \textit{some} noise in the physical layer. It might appear unnatural that the feedback error happens on the initial round, but this could be realized by Alice and Bob somehow knowing the time window in which the feedback noise might occur and starting the protocol during that window. Our constraint on feedback noise is reasonable as well, that any physical Eve-to-Bob channel should not have infinite capacity.

One might say that in order to make the model more physically relevant, there should be noise in the feedback channel to Alice as well. Indeed, it is this very noiselessness that is key to the result that we obtain. The channel to Alice is a noiseless CV classical channel, which has an~\textit{infinite} capacity, while we constrained the channel to Eve to have \textit{finite} capacity. It is this infinite difference that allows us to get a loss unlimited rate~\footnote{Note that this is an infinite difference in capacity between the channels from Bob to Alice and Eve, and we are trying to achieve forward private communication from Alice to Bob. Thus it is not clear a priori that this should be possible.}. Hence, it is unlikely that we can achieve loss unlimited rates if we allow noise in the channel to Alice as well. Indeed, it is known that linear schemes such as Schalkwijk-Kailath cannot achieve any positive rates in this regime~\cite{kim2007gaussian}. 
However, this does not undermine our thesis that~\eqref{eq:shannonHom} indicates issues with the usual assumption of noiselessness of the public side channel. The fact that a classical noiseless CV channel has infinite capacity introduces significant unphysical features to the model and makes achievable rates strongly dependent on the model. For by just adding an infinitesimal amount of noise to the eavesdropper's copy of the initial round, one single step towards physical relevance, we obtain a loss unlimited rate instead of a loss limited one.

Now, it is somewhat queer that we can achieve a markedly different rate even if we allow Eve to have noiseless access to the feedback after the initial round. For an intuition for why this works, we observe that the Schalkwijk-Kailath protocol can be summarized as follows: Bob on the initial round receives the message with additive noise $N_0$ and subsequently learns the value of $N_0$ to exquisite precision. Since Eve has \emph{additional noise} to what Bob received on the initial round, even if Eve can also learn $N_0$ to exquisite precision, the accuracy with which she can recover the message is still limited by the additional noise. Indeed, we can even \textit{further alter} our communication model so as to give Eve $N_0$ exactly as done by~\eqref{eq:privacyBound}, thereby allowing her to completely simulate the protocol for rounds $i>0$, and our result would still hold. We can also give her a copy $X_0'$ of the quantum system~\footnote{Note that this is \textit{not} the classical random variable that is encoded into the quantum system, but the encoded quantum system itself.} $X_0$ that Alice sends on round $i=0$. This would simply replace $E_0$ in the privacy analysis by $X_0' E_0$, whose entropy we can again bound via subadditivity and then by bounds on photon number. Thus, we can alter our communication model so that Eve is \textit{strengthened}, and our result still holds. This reflects again that our protocol essentially depends on that initial round of feedback where there is an infinite difference in capacities.

At first glance the fact that we only need noise on the initial round appears to have implausibly strong implications. One might think our result holds in the usual feedback-assisted private communication model but where Alice and Bob have a pre-shared secret key which they can use to simulate the noise. For instance, by performing bitwise addition with the pre-shared key, they can simulate an additive noise channel with noise $S_0$ on the initial round of feedback. However, since any reasonable pre-shared key has finite Shannon entropy, $S_0$ is discrete. Then, the capacity of the feedback channel would be infinite, which precludes us from bounding the information leaked to Eve.
Hence it does not seem sufficient to have pre-shared secret key to obtain our result.

Lastly, it is worth looking at the closely related communication models of two way private communication, secret key agreement, and quantum key distribution. If we make the same alteration, that is, introduce noise in the classical feedback in the initial round, Bob can simply send an arbitrarily long locally generated random bit string to Alice during that round. Alice will then receive this key noiselessly, while Eve will only receive a noisy copy whose mutual information with the key is finite since the capacity of the noisy feedback channel is finite. Hence, in these models the infinite difference in capacities trivially leads to an infinite rate. This again shows the dependence of the rates to the model considered and suggests that in general we should more seriously address blatant unphysical features in the communication model.

~

\noindent \textbf{Acknowledgements.} 
DD is supported by the Stanford Graduate Fellowship and the National Defense Science and Engineering Graduate Fellowship. SG was supported by the Communications and Networking with Quantum Operationally-Secure Technology for Maritime Deployment (CONQUEST) program funded by the Office of Naval Research (ONR) under a Raytheon BBN Technologies prime contract \# N00014-16-C-2069. 
We would like to thank Patrick Hayden and Tsachy Weissman for valuable discussions and feedback. DD would like to thank God for all of His provisions.

\bibliography{SKQKD}

\cleardoublepage


\appendix
\section{Error Analysis of the Schalkwijk-Kailath Protocol}
\label{appx:SK}

We reproduce for completeness the explicit error analysis of the Schalkwijk-Kailath protocol~\cite{schalkwijk1966coding} as given in~\cite{el2011network} but with our notation. Consider the vector $\mathbf{N} \equiv (N_0, \dots, N_n)$. This is clearly jointly Gaussian since $N_i$ are independent and Gaussian. We claim that  $(N_0 , Y_1, \dots, Y_n)$ is also jointly Gaussian where $Y^n$ are mutually independent and $Y_i \sim \mathcal{N}(0,N_S+ \sigma^2)$. 
\begin{proof}
We proceed by induction. For the base case, we first note that $Y_1 = \gamma_1 N_0 + N_1$, so we obtain $(N_0, Y_1 , N_2 , \dots N_n)$ by a linear transformation, so it's jointly Gaussian with $\E[Y_1] = 0$. Furthermore, $\E[Y_1^2] = \gamma_1^2 +\sigma^2 = N_S+\sigma^2$ since $N_0, N_1$ are independent.

Now we consider the inductive case. Assume $(N_0, Y_1, \dots , Y_k , N_{k+1}, \dots, N_n)$ is jointly Gaussian and $Y^k$ are mutually independent and $Y_i \sim N(0,N_S+\sigma^2)$. Then, $(N_0, Y_1, \dots , Y_k)$ is jointly Gaussian, so $\E[N_0 \vert Y^k]$, which is also the minimum mean squared error (MMSE) estimate, is affine in $Y^k$. Since $Y^k$ and $N_0$ both have mean zero, it is actually linear. Thus, $X_{k+1} = \gamma_{k+1} (N_0 - \E[N_0 \vert Y^k])$ is linear in $N_0, Y^k$. Hence, $(Y_1, \dots, Y_{k+1} = X_{k+1} +N_{k+1})$ is jointly Gaussian. By the orthogonality principle for linear MMSE estimates, $X_{k+1}$ is independent of $Y_i$ for $i \in [1:k]$. Furthermore, $N_{k+1}$ is also independent of $Y_i$. Since $Y^k$ are mutually independent, we conclude that $(Y_1, \dots , Y_{k+1})$ is jointly Gaussian and uncorrelated. Thus, they are mutually independent. It is also clear that $\E[Y_{k+1}] = 0 $. Furthermore, $\E[Y_{k+1}^2] = N_S+\sigma^2$ since $\E[X_{k+1}^2] = N_S$ by construction and $X_{k+1},N_{k+1}$ are independent. Finally, $(N_0, Y_1, \dots, Y_{k+1} , N_{k+2}, \dots, N_n)$ is jointly Gaussian, so we're done.
\end{proof}

We next expand $I(N_0 ; Y^n)$ in two different ways. First,
\begin{align*}
  I(N_0 ; Y^n) &= \sum_{i=1}^{n} I(N_0 ; Y_i \vert Y^{i-1})\\
  &= \sum_{i=1}^{n} h(Y_i \vert Y^{i-1}) - h(Y_i \vert N_0, Y^{i-1})\\
  &= \sum_{i=1}^{n} h(Y_i) - h(N_i \vert N_0, Y^{i-1}) \\
  &= \sum_{i=1}^{n} h(Y_i) - h(N_i) \\
  & = \frac{n}{2} \log\left( 1+ \frac{N_S}{\sigma^2}\right) \\
  & = n P_H
\end{align*}
where the third equality holds because $Y^n$ are mutually independent and $Y_i$ is a function of $Y^{i-1}$, $N_0$, and $N_i$. $P_H$ is the capacity of the AWGN channel with power $N_S$ and noise $\sigma^2$ given in~\eqref{eq:shannonHom}. The second way to calculate this gives
\begin{align*}
  I(N_0 ; Y^n) & = h(N_0) - h(N_0 \vert Y^n) \\
  & = \frac{1}{2} \log(2\pi e \Var[N_0]) - \frac{1}{2} \log(2\pi e \Var[N_0 \vert Y^n])\\
  & = -\frac{1}{2} \log\frac{\Var[N_0 \vert Y^n]}{\sigma^2},
\end{align*}
where the second equality follows from the fact that $N_0 \vert Y^n = y^n$ is Gaussian with variance independent of $y^n$ by the joint Gaussianity of $(N_0, Y_1, \dots, Y^n)$. Note that $\Var[N_0 \vert Y^n]$ is a random variable, but since it does not depend on $y^n$, we can identify it with the value it takes on almost surely.

We conclude from the two ways to write the mutual information that $\Var[N_0 \vert Y^n] = \sigma^2 2^{-2n P_H}$. Now, $\Theta_n = Y_0 - \E(N_0 \vert Y^n)$ is Gaussian since it's linear in $Y_0^n$. Furthermore,
\begin{align*}
  \Var{\Theta_n} & = \E[(N_0 - \E[N_0 \vert Y^n])^2] \\
  & = \E[\E[(N_0 - \E[N_0 \vert Y^n])^2 \vert Y^n]] \\
  & = \Var[N_0 \vert Y^n],
  \label{}
\end{align*}
where the second equality follows by the law of iterated expectation. Thus, $\Theta_n \sim \mathcal{N}( \theta(m), \sigma^2 2^{-2n P_H})$. 

We make a decoding error only if $\abs{\Theta_n-\theta(m)} > \sqrt{N_S} 2^{-nR}$. Hence, $p_e \le 2 Q(2^{n(P_H- R)} \sqrt{N_S/\sigma^2})$, where 
\begin{equation*}
  Q(x)\equiv \int_x^\infty \frac{1}{\sqrt{2\pi}} e^{-t^2/2} dt.
  \label{}
\end{equation*}
Now, we know that for $x \ge 1$, $Q(x) \le \frac{1}{\sqrt{2\pi}} e^{-x^2/2}$~\cite{durrett2010probability}, so if $R< P_H$ and $n$ is large enough, 
\begin{equation*}
  p_e \le \sqrt{\frac{2}{\pi}} \exp\left( -\frac{2^{2n(P_H-R)}N_S}{2\sigma^2} \right)
  \label{}
\end{equation*}
Note that we used the channel $n+1$ times, so the rate we achieve is actually $\frac{n}{n+1}R$. 

\end{document}